\documentclass[10pt,a4paper,superscriptaddress,nofootinbib]{revtex4}

\usepackage[T1]{fontenc}
\usepackage[ansinew]{inputenc}
\usepackage{amsmath}
\usepackage{amssymb}
\usepackage{amsthm}
\usepackage{hyperref}

\usepackage{amsmath}
\usepackage{amssymb}

\newtheorem{p}{Proposition}
\newtheorem{cor}{Corollary}
\newtheorem{de}{Definition}
\newtheorem{lem}{Lemma}

\DeclareMathOperator{\tr}{tr}
\DeclareMathOperator{\id}{id}
\DeclareMathOperator{\TPCPM}{TPCPM}

\def\roh{\hat{\rho}}
\def\rooh{\hat{\rho}_{AB}}
\def\sih{\hat{\sigma}}

\def\roo{\rho_{AB}}
\def\roBx{\rho_B^x}
\def\rot{\tilde{\rho}}
\def\root{\tilde{\rho}_{AB}}
\def\sii{\sigma_B}
\def\siit{\tilde{\sigma}_B}

\def\sitC{\tilde{\sigma}_C}

\def\mcor{\vert \Psi_{AB} \rangle \langle \Psi_{AB} \vert }
\def\Hmin{H_{\mathrm{min}}(\rho_{AB}|\sigma_B)}
\def\Hmink{H_{\mathrm{min}}(\rho^k_{AB}|\sigma^k_B)}
\def\HminB{H_{\mathrm{min}}(\rho_{AB}|B)}
\def\HmaxB{H_{\mathrm{max}}(\rho_{AB}|B)}
\def\HminBk{H_{\mathrm{min}}(\rho^k_{AB}|B)}
\def\HmaxBk{H_{\mathrm{max}}(\rho^k_{AB}|B)}
\def\HminBt{H_{\mathrm{min}}(\tilde{\rho}_{AB}|B)}
\def\HmaxBt{H_{\mathrm{max}}(\tilde{\rho}_{AB}|B)}

\def\sHmin{H_{\mathrm{min}}^{\epsilon}(\rho_{AB}|B)}
\def\sHmax{H_{\mathrm{max}}^{\epsilon}(\rho_{AB}|B)}
\def\sinf{ \inf_{\tilde{\rho}_{AB} \in \mathcal{B}^{\epsilon}(\rho_{AB})}}
\def\ssup{ \sup_{\tilde{\rho}_{AB} \in \mathcal{B}^{\epsilon}(\rho_{AB})}}
\def\ssupp{ \sup_{{\sigma}_{AB} \in \mathcal{B}^{\epsilon}(\rho_{AB})}}

\def\ball{\mc{B}^{\epsilon} }

\def\hahb{H_A \otimes H_B}
\def\ha{H_A}
\def\hb{H_B}
\def\mc{\mathcal}

\begin{document}

\title{Min- and Max-Entropy in Infinite Dimensions}

\author{Fabian Furrer}
\email{fabian.furrer@itp.uni-hannover.de}
\affiliation{Institute for Theoretical Physics, Leibniz Universit\"at Hannover, 30167 Hannover, Germany}

\author{Johan {\AA}berg}
\email{jaaberg@phys.ethz.ch}
\affiliation{Institute for Theoretical Physics, ETH Zurich,
             8093 Zurich, Switzerland}

\author{Renato Renner}
\email{renner@phys.ethz.ch} \affiliation{Institute for Theoretical Physics, ETH Zurich,
             8093 Zurich, Switzerland}

\begin{abstract}
We consider an extension of the conditional min- and max-entropies to infinite-dimensional separable Hilbert spaces. We show that these satisfy characterizing properties known from the finite-dimensional case, and retain information-theoretic operational interpretations, e.g., the  min-entropy as maximum achievable quantum correlation, and the max-entropy as decoupling accuracy.
We furthermore generalize the smoothed versions of these entropies and prove an infinite-dimensional  quantum asymptotic equipartition property. To facilitate these generalizations we show that the min- and max-entropy can be expressed in terms of convergent sequences of finite-dimensional min- and max-entropies, which provides a convenient technique to extend proofs from the finite to the infinite-dimensional setting.
\end{abstract}

\maketitle

\section{Introduction}

Entropy measures are fundamental to information theory. For example, in classical
information theory a central role is played by the Shannon entropy \cite{Shannon}
and in quantum information theory by the von Neumann entropy. Their usefulness partially stems from the fact that they have several convenient mathematical properties (e.g. strong subadditivity) that facilitate a `calculus' of information and uncertainty. Indeed, entropy measures can even be characterized axiomatically in terms of such properties \cite{Renyi}. However, equally important for their use in information theory is the fact that they are related to operational quantities. This means that they characterize the optimal efficiency by which various information-theoretic tasks can be solved. One example of such a task is source coding, where one considers a source that randomly outputs data according to some given probability distribution. The question of interest is how much memory is needed in order to store and faithfully regenerate the data. Another example is channel coding, where the aim is to reliably transmit information
over a channel. Here we ask how many bits (or qubits in the quantum
case) one can optimally transmit per use of the channel \cite{Shannon,ChannelCoding,Schuhmacher}.

The operational relevance of Shannon and von Neumann entropy is normally limited to the case when one considers the asymptotic limit over infinitely many instances of a random experiment, which are independent and identically distributed (iid) or can be described by a Markov process. In the case of source coding this corresponds to assuming an iid repetition of the source. In the limit of infinitely many such repetitions, the average number of bits one needs to store per output is given by the Shannon entropy of the distribution of the source \cite{Shannon}. In the general case, where we have more complicated types of correlations, or where we only consider finite instances, the role of the Shannon or von Neumann entropies appears to be taken over by other measures of entropy, referred to as the smooth min- and max-entropies  \cite{RennerPhD}.
For example, in \cite{ChannelCodingMaxEntropy,RenesRenner} it was found that the smooth max-entropy characterizes one-shot data compression, i.e., when we wish to compress a single output of an information source. Furthermore,
in \cite{ChannelCodingMinEntropy} it was proved that in one single use of a classical channel, the transmission can be characterized by the  difference between a smooth min- and max-entropy. The von Neumann entropy of a state can be regained  via the quantum asymptotic equipartition property (AEP) \cite{RennerPhD,QuantumAEP}, by applying these measures to asymptotically many iid repetitions of the
state. This allows us to derive properties of the von Neumann entropy from the smooth min- and max-entropies; a technique that has been used for an alternative proof of the quantum reverse Shannon theorem \cite{Reverse}, and to derive an entropic uncertainty relation \cite{Uncertainty}.
The min- and max-entropies furthermore generalize the spectral entropy rates \cite{InformationSpectrum}  (that are defined in an asymptotic sense) which themselves have been introduced as generalizations of the Shannon entropy \cite{Han,HanVerdu}.
Closely related quantities  are  the  relative min- and max-entropies \cite{parent}, which have been applied to entanglement theory \cite{entang1,entang2} as well as channel capacity \cite{capacity}.

So far, the investigations of the operational relevance and properties of the min- and max-entropy and their smoothed versions have been almost exclusively focused on quantum systems with finite-dimensional Hilbert spaces. Here we consider the min- and max-entropy in infinite-dimensional separable Hilbert spaces. Since the modeling in vast parts of quantum physics is  firmly rooted in infinite-dimensional Hilbert spaces, it appears that such a generalization is crucial for the application of these tools. For example, it has recently been shown that the smooth min- and max-entropies are the relevant measures of entropy in certain statistical mechanics settings \cite{Oscar,Lidia}. An extension of these ideas to, e.g., quantized classical systems, would require an infinite-dimensional version of the min- and max-entropy.
Another example is quantum key distribution (QKD), where in the
finite-dimensional case the smooth min-entropy bounds the length of the
secure key that can be extracted from an initial raw key \cite{RennerPhD}. The
generalization to infinite dimensions has therefore direct relevance for continuous
variable QKD (for references see, e.g., Section II.D.~3  of \cite{Scarani}).
In such a scheme one uses the quadratures of the electromagnetic field to establish a secret key (as opposed to other schemes that use, e.g.,  the polarization degree of freedom of single photons).  Since such QKD methods are based on the generation of coherent states and measurement of quadratures, it appears rather unavoidable to use infinite-dimensional Hilbert spaces to model the states of the field modes.
Beyond the obvious application to continuous variable quantum key distribution, one can argue that there are several quantum cryptographic tasks that today are analyzed in finite-dimensional settings, which strictly speaking would require an analysis in infinite-dimensions, since there is in general no reason to assume the Hilbert spaces of the adversary's systems to be finite.

As indicated by the above discussion, an extension of the min- and max-entropies to an infinite-dimensional setting does not only require that we can reproduce known mathematical properties of these measures, but also that we should retain their operational interpretations. A complete study of this two-fold goal would bring us far beyond the scope of this work. However, here we pave the way for this development by introducing an infinite-dimensional generalization of the min- and max-entropy, and demonstrating  a collection of `core' properties and operational interpretations. In particular, we derive (under conditions detailed below) a quantum AEP for a specific choice of an infinite-dimensional conditional von Neumann entropy.
On a more practical level we introduce a technique that facilitates the extension of results proved for the finite-dimensional case to the setting of separable Hilbert spaces. More precisely, we show that the conditional min- and max-entropies for infinite-dimensional states can be
 expressed as limits of entropies obtained by finite-dimensional truncations of the original state (Proposition~\ref{p:reduction of Hmin to finite dim}). This turns out to be a convenient tool for generalizations, and we illustrate this on the various infinite-dimensional extensions that we consider.

The $\epsilon$-smoothed min-and max-entropies are defined in terms of the `un-smoothed' ($\epsilon=0$) min- and max-entropies (which we simply refer to as `min- and max-entropy').  In Section~\ref{subsec:def} we extend these `plain' min- and max-entropies to separable Hilbert spaces.
 Section~\ref{section:reduction of cond entropies} contains the main technical tool, Proposition~\ref{p:reduction of Hmin to finite dim}, by which the infinite-dimensional min- and max-entropies can be expressed as limits of sequences of finite-dimensional entropies. The proof of Proposition~\ref{p:reduction of Hmin to finite dim} is given in Appendix~\ref{app:proofprop1}.
In Section~\ref{section:properties of min- and max-entropy} we consider properties of the min- and max-entropy, e.g.,  additivity and the data processing inequality. Section~\ref{chapter:operational interpretation} focuses on the generalization of operational interpretations.
In Section~\ref{chapter:smooth entropy} we consider the extension of the $\epsilon$-smooth min- and max-entropies, for $\epsilon>0$.
In Section~\ref{AEP} we bound the smooth min- and max-entropy of an iid state on a system $A$ conditioned on a system $B$ in terms of the conditional von Neumann entropy (Proposition \ref{prop:AEP lower bound}). This result relies on the assumption that $A$ has finite von Neumann entropy. If $A$ furthermore has a finite-dimensional Hilbert space (but the Hilbert space of $B$ is allowed to be separable) we prove that these smooth entropies converge to the conditional von Neumann entropy  (Corollary \ref{cor:AEP}), which corresponds to a quantum AEP.
 The paper ends with a short summary and outlook in Section~\ref{concl}.

\section{\label{chapter:conditioned entropies} Min- and Max-Entropy}

\subsection{\label{subsec:def}Definition of the conditional min- and max-entropy}

Associated to each quantum system is a Hilbert space $H$, which we assume to be separable in all that follows. We denote the bounded operators by $\mc{L}(H)=\{A:H \rightarrow H \ |\ \Vert A\Vert  < \infty\}$, where $\Vert A\Vert   = \sup_{\Vert \psi\Vert  = 1}\Vert A|\psi\rangle\Vert $ is the standard operator norm. Among these, the trace class operators satisfy the additional feature of having a finite trace norm $\Vert T\Vert_1 := \tr|T| = \tr\sqrt{T^{\dagger}T}$. The set of trace class operators is denoted by $\tau_1(H):= \{ T \in \mc{L}(H)| \ \Vert T\Vert_1 < \infty \}$.

We consider states which can be represented as density operators, i.e., normal states \cite{Bratteli Robinson}, and denote the set of all these states as $\mc{S}(H):=\{\rho \in \tau_1(H)| \ \rho \geq 0, \Vert \rho\Vert_1 =1 \}$. It is often convenient to allow non-normalized density operators, which form the positive cone $\tau^+_1(H)\subset\tau_1(H)$ consisting of all non-negative trace class operators.

We define the conditional min- and max-entropy of bipartite quantum systems  analogously to the finite-dimensional case \cite{OperationalMeaning}.\footnote{Max-entropy as we define it in Eq.~(\ref{def,eq3:min/max-entropy}) is related to the R\'enyi 1/2-entropy (see Section~\ref{subsect:entropy of pure states versus unconditioned entropy} or \cite{OperationalMeaning,OnTheSmoothing}). In the original definition \cite{RennerPhD} max-entropy was defined in terms of the R\'enyi 0-entropy.}

\begin{de} \label{def:min/max-entropy}Let $H_A$ and $H_B$ be separable Hilbert spaces and $\roo \in \tau_1^+(\hahb)$. The min-entropy of $\rho_{AB}$ conditioned on $\sigma_{B} \in \tau_1^+(\hb)$ is defined by
\begin{equation}\label{def,eq1:min/max-entropy}
H_{\mathrm{min}}(\rho_{AB}|\sigma_{B}) := -\log \inf \{\lambda \in \mathbb{R} | \lambda \id_{A} \otimes \sigma_{B}  \geq \rho_{AB} \},
\end{equation}
where we let $H_{\mathrm{min}}(\roo|\sii) := -\infty$ if the condition $\lambda \id_{A} \otimes \sigma_{B}  \geq \rho_{AB}$ cannot be satisfied for any $\lambda \in \mathbb{R}$.
Moreover, we define the min-entropy of $\roo$ conditioned on B by
\begin{equation}\label{def,eq2:min/max-entropy}
 H_{\mathrm{min}}(\rho_{AB}|B) := \sup_{\sigma_{B} \in \mc{S}(\hb)} H_{\mathrm{min}}(\rho_{AB}|\sigma_{B}).
\end{equation}
The max-entropy of $\rho_{AB}$ conditioned on B is defined as the dual of the min-entropy
\begin{equation}\label{def,eq3:min/max-entropy}
 H_{\mathrm{max}}(\rho_{AB}|B) := -H_{\mathrm{min}}(\rho_{AC}|C),
\end{equation}
where $\rho_{ABC}$ is a purification of $ \rho_{AB}$.
\end{de}
In the definition above, and in all that follows, we let ``$\log$'' denote the binary logarithm. The reduction of a state to a subsystem is indicated by the labels of the Hilbert space,  e.g., $\rho_A = \tr_C\rho_{AC}$. Note that the max-entropy $ H_{\mathrm{max}}(\rho_{AB}|B)$ as defined in (\ref{def,eq3:min/max-entropy}) is independent of the choice of the purification $\rho_{ABC}$, and thus well-defined. This follows from the fact that two purifications can only differ by a partial isometry on the purifying system, and the min-entropy $H_{\mathrm{min}}(\rho_{AC}|C)$ is invariant under these partial isometries on subsystem C.

The two optimizations in the definition of $H_{\mathrm{min}}(\rho_{AB}|B)$, in Eqs.~(\ref{def,eq1:min/max-entropy}) and (\ref{def,eq2:min/max-entropy}), can be combined into
\begin{equation}\label{eq:equiv def of the min-entropy}
H_{\mathrm{min}}(\rho_{AB}|B) = -\log \big(\inf \{ \tr\siit \ | \ \siit \in \tau^+_1(H_B), \\ \id_A \otimes \siit \geq \roo\} \big).
\end{equation}
For convenience we introduce the following two quantities:
\begin{eqnarray}
\label{lambdadef1}
\Lambda(\roo|\sii) & := &   2^{- H_{\mathrm{min}}(\rho_{AB}|\sii)} =\inf \{\lambda \in \mathbb{R} | \lambda \id_{A} \otimes \sigma_{B}  \geq \rho_{AB} \},\\
 \label{lambdadef2} \Lambda(\roo|B)  &:= & 2^{- H_{\mathrm{min}}(\rho_{AB}|B)} = \inf \{ \tr\siit \ | \ \siit \in \tau^+_1(H_B),  \id_A \otimes \siit \geq \roo\}.
\end{eqnarray}

\subsection{\label{section:reduction of cond entropies} Finite-dimensional approximations of min- and max-entropies}

In this section we present the main result, Proposition~\ref{p:reduction of Hmin to finite dim}, that provides a method to express the conditional min- and max-entropy as a limit of min- and max-entropies of finite-dimensional systems. The rough idea is to choose sequences $\{P^A_k\}_{k=1}^{\infty}$ and $\{P^B_k\}_{k=1}^{\infty}$ of projectors \footnote{With ``projector'' we intend a bounded operator $P$ such that $P^2 = P$ and $P^{\dagger} = P$, which in the mathematics literature usually is referred to as an ``orthogonal projector''.}
onto finite-dimensional subspaces $U_k^A \subset H_A$ and $U_k^B \subset H_B$, respectively, both converging to the identity.
Then we define a sequence of non-normalized states as $\roo^k = (P^A_k \otimes P^B_k) \roo (P^A_k \otimes P^B_k)$. The min- or max-entropy of $\roo^k$ can now be treated as if the underlying Hilbert space would be $U_k^A \otimes U_k^B$ (Lemma~\ref{I:id-P}), and therefore finite-dimensional. Proposition~\ref{p:reduction of Hmin to finite dim} shows that, as $k\rightarrow \infty$, these finite-dimensional entropies approach the desired infinite-dimensional entropy. As we will see, this provides a convenient method to extend properties from the finite to the infinite setting.

When we say that an operator sequence $Q_{k}$ converges to $Q$ in the weak operator topology we intend that $\lim_{k\rightarrow 0} \langle \chi |Q-Q_{k}|\psi\rangle = 0$  for all $|\chi\rangle,|\psi\rangle\in H$. The sequence converges in the strong operator topology if $\lim_{k\rightarrow 0} \Vert (Q-Q_{k})|\psi\rangle\Vert = 0$ for all $|\psi\rangle\in H$.

\begin{de} \label{def:projected states}
Let $\{P^A_k\}_{k\in\mathbb{N}} \subset \mc{L}(H_A)$, $\{P^B_k\}_{k\in\mathbb{N}} \subset \mc{L}(H_B)$ be sequences of projectors such that for each $k \in \mathbb{N}$ the projection spaces $U_k^A\subset H_A$, $U_k^B\subset H_B$ of $P_k^A$, $P_k^B$ are finite-dimensional, $P_k^A \leq P_{k'}^A$ and $P_k^B \leq P_{k'}^B$ for all $k \leq k'$, and $P_k^{A}$, $P_k^B$ converge in the weak operator topology to the identity. We refer to such a sequence $(P_k^A,P_k^B)$ as a generator of projected states. For $\roo \in \mc{S}(\hahb)$ we define the (non-normalized) states
\begin{equation}
\label{defprojectedstates}
\roo^k := (P^A_k \otimes P^B_k) \roo (P^A_k \otimes P^B_k),
\end{equation}
which we call the projected states of $\roo$ relative to $(P_k^A,P_k^B)$. Moreover, we refer to
\begin{equation}
\label{snklv}
\rooh^k := \frac{\roo^k}{ \tr\roo^k}
\end{equation}
as the normalized projected states of $\roo$ relative to $(P_k^A,P_k^B)$.
\end{de}

Note that a sequence of projectors that converges in the weak operator topology to the identity also converges in the strong operator topology to the identity. As a matter of convenience, we can thus in all that follows regard the
generators of projected states to converge in the strong operator topology. One may also note that the sequence of projected states $\roo^k$ (as well as the normalized projected states $\rooh^k$) converges to $\roo$ in the trace norm (see Corollary~\ref{dfnbkl}  in Appendix~\ref{section:Technical Lemmas}). The normalized projected states in Eq.~(\ref{snklv}) are of course only defined if  $\tr\roo^k \neq 0$. However,   this is true for all sufficiently large $k$ due to the trace norm convergence to $\roo$.

\begin{p}\label{p:reduction of Hmin to finite dim}
For $\roo\in\mc{S}(\hahb)$, let $\{\roo^k\}_{k\in\mathbb{N}}$ be the projected states of $\roo$ relative to a generator $(P_k^A,P_k^B)$, and  $\roh^k_{AB}$ the corresponding normalized projected states. Furthermore, let $\sii \in \mc{S}(\hb)$ and define the operators $\sigma_B^k := P^B_k\sigma_B P^B_k$ and $\sih_B^k :=\tr(\sigma_B^k)^{-1}\sigma^k_B$.  Then, the following three statements hold.
\begin{equation} \label{p,eq1:reduction of Hmin to finite dim}
\Hmin = \lim_{k\rightarrow \infty} \Hmink =\lim_{k\rightarrow \infty} H_{\mathrm{min}}\big(\roh^k_{AB}|\sih^k_B\big),
\end{equation}
and the infimum in Eq.~(\ref{def,eq1:min/max-entropy}) is attained if $\Hmin$ is finite.
\begin{equation}\label{p,eq2:reduction of Hmin to finite dim}
\HminB = \lim_{k\rightarrow \infty} H_{\mathrm{min}}(\roo^k|B_k) = \lim_{k\rightarrow \infty} H_{\mathrm{min}} \big(\roh^k_{AB}|B_k\big),
\end{equation}
and the supremum in Eq.~(\ref{def,eq2:min/max-entropy}) is attained if $\HminB$ is finite.
\begin{equation}\label{p,eq3:reduction of Hmin to finite dim}
\HmaxB = \lim_{k \rightarrow \infty} H_{\mathrm{max}}(\roo^k|B_k) = \lim_{k\rightarrow \infty} H_{\mathrm{max}} \big(\roh^k_{AB}|B_k\big).
\end{equation}
Here,  $B_k$ denotes the restriction of system $B$ to the projection space $U_k^B$ of $P_k^B$.
\end{p}
The proof of this proposition is found in Appendix~\ref{app:proofprop1}.
When we say that the infimum in (\ref{def,eq1:min/max-entropy}) is attained, it means that there exists a finite $\lambda'$ such that $ \lambda' \id_{A} \otimes \sigma_{B}  - \rho_{AB} \geq 0$ and $H_{\mathrm{min}}(\rho_{AB}|\sigma_{B}) = -\log \lambda'$. Similarly, that the supremum in (\ref{def,eq2:min/max-entropy}) is attained, means that there exists a $\sigma'_{B} \in \tau_1^+(H_B)$ satisfying $\id \otimes \sigma'_B \geq \roo$ such that  $H_{\mathrm{min}}(\rho_{AB}|B) =  H_{\mathrm{min}}(\rho_{AB}|\sigma'_{B})$.

Given the above proposition, a natural question is if $\HminB$ and $\HmaxB$ are trace norm continuous in general. In the finite-dimensional case \cite{OnTheSmoothing} it is known that these entropies are continuous with a Lipschitz constant depending on the dimension of $H_A$. However, the following example shows that they are in general not continuous in the infinite-dimensional case. Let $\{|k\rangle\}_{k= 0,1,\ldots}$ be an arbitrary orthonormal basis of the Hilbert space $H_A$. For each $n= 1,2,\ldots$ let
\begin{equation}
\rho_{n} = (1-\frac{1}{n})|0\rangle\langle 0| + \frac{1}{n^{2}}\sum_{k=1}^{n}|k\rangle\langle k|.
\end{equation}
One can see that $\rho_n$ converges in the trace norm to $|0\rangle\langle0|$ as $n\rightarrow\infty$, while
$\lim_{n\rightarrow\infty}H_{\mathrm{max}}(\rho_n) = 2$, and $H_{\mathrm{max}}(|0\rangle\langle 0|) = 0$. Hence, the max-entropy is not continuous. ($H_{\mathrm{max}}(\rho)$ without conditioning means that we condition on a trivial subsystem $B$. See Eq.~(\ref{cor,eq1:unconditioned max-entropy}).) The duality, Eq.~(\ref{def,eq3:min/max-entropy}), yields an example also for the min-entropy.


\section{\label{section:properties of min- and max-entropy} Properties of Min- and Max-Entropy}
\subsection{Additivity and the data processing inequality}
Proposition~\ref{p:reduction of Hmin to finite dim} can be used as a tool to generalize known finite-dimensional results to the infinite-dimensional case. A simple example is the ordering property \cite{QuantumAEP}
\begin{equation}\label{orderingHminHmax}
H_{\mathrm{min}}(\roo|B) \leq H_{\mathrm{max}}(\roo|B),
\end{equation}
which is obtained by a direct application of Proposition~\ref{p:reduction of Hmin to finite dim}. Another example is the additivity, which in the finite-dimensional case was proved in
\cite{RennerPhD}. A direct generalization of the proof techniques they employed appears rather challenging, while Proposition~\ref{p:reduction of Hmin to finite dim} makes the generalization straightforward.
\begin{p} \label{p:additivity}
Let $\roo \in \mc{S}(\hahb)$ and $\rho_{A'B'} \in \mc{S}(H_{A'} \otimes H_{B'}) $ for $H_A, H_{A'}$, $H_B$, and $H_{B'}$ separable Hilbert spaces. Then, it follows that
\begin{align}  \label{Hminadd} H_{\mathrm{min}}(\roo \otimes \rho_{A'B'}| BB')  & =  H_{\mathrm{min}}(\roo|B) + H_{\mathrm{min}}(\rho_{A'B'}|B'), \\
    \label{Hmaxadd} H_{\mathrm{max}}(\roo \otimes \rho_{A'B'}|BB') & = H_{\mathrm{max}}(\roo|B) + H_{\mathrm{max}}(\rho_{A'B'}|B').
\end{align}
\end{p}
The proof is a simple application of the approximation scheme in Proposition~\ref{p:reduction of Hmin to finite dim} combined with Lemma~\ref{nvdakj} and the finite-dimensional version of Proposition~\ref{p:additivity}, and therefore omitted.

For the sake of completeness we note that the data processing inequalities \cite{RennerPhD} also hold in the infinite-dimensional setting. In this case, however, there is no need to resort to Proposition~\ref{p:reduction of Hmin to finite dim}, as the proof in \cite{RennerPhD} can be generalized directly.
\begin{p} \label{p:strsubadditivity}
Let $\rho_{ABC} \in \tau_{+}(\hahb\otimes H_C)$ for separable Hilbert spaces $H_A$, $H_B$ and $H_C$. Then, it follows that
\begin{eqnarray}
\label{p,eq2:strong subadd of Hmin} & &  H_{\mathrm{min}}(\rho_{ABC}|BC) \leq H_{\mathrm{min}}(\rho_{AB}|B),  \\
 \label{p:strong subadd of Hmax}   & &  H_{\mathrm{max}}(\rho_{ABC}|BC) \leq H_{\mathrm{max}}(\rho_{AB}|B).
\end{eqnarray}
\end{p}
The data processing inequalities can be regarded as the min- and max-entropy counterparts of the strong subadditivity of the von Neumann entropy (and are sometimes directly referred to as ``strong subadditivity'').
One reason for this is that the standard formulation of the strong subadditivity of von Neumann entropy \cite{LiebRuskai1,Lieb,LiebRuskai2}, $H(\rho_{ABC}) + H(\rho_{B})\leq H(\rho_{AB}) + H(\rho_{BC})$,  can be recast in the same form.


\subsection{\label{subsect:entropy of pure states versus unconditioned entropy} Entropy of pure states, and a bound for general states}
Here we briefly consider the fact that the min-entropy can take the value $-\infty$, and  the max-entropy can take the value $+\infty$. For this purpose we discuss the special case of pure states, as well as the case of no conditioning (i.e., if there is no subsystem $B$). Based on this we obtain a general bound which says that the conditional min- and max-entropies of a state $\roo$ are finite if the operator $\sqrt{\rho_A}$ is trace class. Moreover it turns out that the min-entropy cannot attain the value $+\infty$, while the max-entropy cannot attain $-\infty$.
\begin{lem}
\label{cor:min-entropy for purestates}
The min-entropy of $\roo=\vert \psi \rangle \langle \psi \vert$, where $\vert\psi\rangle \in \hahb$, is given by
\begin{equation}\label{cor,eq1:min-entropy for purestates}
H_{\mathrm{min}}(\roo|B) = - 2 \log\tr\sqrt{\rho_A}.
\end{equation}
\end{lem}
From this lemma we can conclude that $H_{\mathrm{min}}(\roo|B)$ is finite if and only if $\sqrt{\rho_{A}}$ is trace class. Otherwise $H_{\mathrm{min}}(\roo|B) = -\infty$.
If the Schmidt decomposition \cite{Convertability} of $\psi$ is given by $\sum_{k=1}^{\infty} r_k \vert a_k\rangle\vert b_k\rangle$, we have $\tr\sqrt{\rho_A}=\sum_{k=1}^{\infty}r_k$, such that a finite Schmidt rank always implies that $H_{\mathrm{min}}(\roo|B)$ is finite.
Recall that the Schmidt coefficients characterize the entanglement of a pure state, and, roughly speaking, that the more uniformly the Schmidt coefficients are distributed the stronger is the entanglement (see for instance \cite{Convertability}). This suggests that pure states with $H_{\mathrm{min}}(\roo|B) =-\infty$ are entangled in a rather strong sense.
\begin{proof}
Let $\vert \psi \rangle = \sum_{k=1}^{\infty} r_k \vert a_k\rangle\vert b_k\rangle$ be the Schmidt decomposition of $\vert \psi \rangle$, and $\tilde{\sigma}_{B} \in \tau^+_1(\hb)$ any operator that satisfies $\id_A \otimes \siit \geq \roo$. For each $n \in \mathbb{N}$ define $|\chi_{n}\rangle = \sum_{k=1}^{n}|a_{k}\rangle\vert b_{k}\rangle$.
It follows that
\begin{equation*}
 \tr\siit \geq \langle\chi_{n}|\id_A\otimes \siit|\chi_{n}\rangle \geq \langle\chi_{n}|\roo|\chi_{n}\rangle = \big(\sum_{k=1}^{n}r_{k}\big)^{2},
\end{equation*}
and thus, by taking the infimum over all  $\tilde{\sigma}_{B}$ with $\id_A \otimes \siit \geq \roo$, as well as the supremum over all $n$, we find $\Lambda(\rho_{AB}|B)  \geq    (\tr\sqrt{\rho_A})^2$. Especially, we see that if $\tr\sqrt{\rho_A} = +\infty$ then $\Lambda(\roo|B) = +\infty$ (and thus $H_{\mathrm{min}}(\roo|B)=-\infty$).
In the following we assume that $\tr\sqrt{\rho_A}< +\infty$, i.e., $\sqrt{\rho_A} \in \tau^+_1(\ha)$. We show that the lower bound $\Lambda(\roo|B) \geq  (\tr\sqrt{\rho_A})^2$ is attained, by proving that $\siit := \tr(\sqrt{\rho_A})\sqrt{\rho_B}$ satisfies $id_A \otimes \siit \geq \roo$.
By using the Schmidt decomposition of $\psi$ we compute for an arbitrary $\eta \in \hahb$
\begin{align*}
 \langle \eta \vert ( \id \otimes\siit-\roo) \vert \eta\rangle
 = & \tr(\sqrt{\rho_A})\sum_{k,l=1}^{\infty}|c_{k,l}|^{2}r_{l} -\Big|\sum_{k=1}^{\infty}c_{k,k}r_{k}\Big|^{2}  \\
 \geq &  \sum_{l=1}^{\infty}r_l \sum_{k=1}^{\infty}|c_{k,k}|^{2}r_{k} - \Big|\sum_{k=1}^{\infty}c_{k,k}r_{k}\Big|^{2} \geq 0 ,
\end{align*}
where $c_{k,l} = (\langle a_k\vert\langle b_l\vert)\vert \eta \rangle$, and the last step follows from the Cauchy-Schwarz inequality. Hence, $\id_A \otimes \siit - \roo$ is positive and therefore $  \tr(\siit) \geq \Lambda(\roo|B)$. Combined with $\Lambda(\roo|B) \geq  (\tr\sqrt{\rho_A})^2$, we find $H_{\mathrm{min}}(\roo|B) = -\log \Lambda(\roo|B)=- 2 \log\tr\sqrt{\rho_A}$.
\end{proof}

The duality (\ref{def,eq3:min/max-entropy}) allows us to rewrite Lemma~\ref{cor:min-entropy for purestates} by using the unconditional max-entropy. For every $\rho \in \mc{S}(H)$ this yields the quantum $1/2$-R\'enyi entropy (cf. \cite{OperationalMeaning}),
\begin{equation}\label{cor,eq1:unconditioned max-entropy}
 H_{\mathrm{max}}(\rho) = 2\log\tr\sqrt{\rho} = H_{\frac{1}{2}}(\rho),
\end{equation}
if $\sqrt{\rho}$ is trace-class. Otherwise $H_{\mathrm{max}}(\rho)=+\infty$.

The unconditional min-entropy is obtained by conditioning on a trivial subsystem $B$. One can see that
\begin{equation}\label{cor,eq1:unconditioned min-entropy}
H_{\mathrm{min}}(\rho) = -\log \Vert \rho\Vert .
\end{equation}
For a pure state $\roo = \vert \psi \rangle \langle \psi \vert\in \mc{S}(\hahb)$, the max-entropy is given by
\begin{equation}\label{cor,eq1:max-entropy of pure states}
 H_{\mathrm{max}}(\roo|B) = \log \Vert \rho_A\Vert.
\end{equation}
To see this one can apply the duality (\ref{def,eq3:min/max-entropy}) where we purify the pure state $\roo$ with a trivial system $C$,  and next use Eq.~(\ref{cor,eq1:unconditioned min-entropy}).

By combining these facts with the data processing inequality, $H_{\mathrm{min}}(\rho_{ABC}|BC) \leq \HminB \leq H_{\mathrm{min}}(\rho_A)$ and $H_{\mathrm{max}}(\rho_{ABC}|BC) \leq \HmaxB \leq H_{\mathrm{max}}(\rho_A)$, for $\rho_{ABC}$ a purification of $\roo$, we find the following bounds on the min- and max-entropy.
\begin{p}\label{cor:upper bound for min/max-entropy}
For every state $\roo \in S(\hahb)$ it holds that
\begin{eqnarray}
 -2\log\tr\sqrt{\rho_A} \leq &\HminB  & \leq -\log \Vert \rho_A\Vert,   \label{cor,eq1:upper bound for min-entropy} \\
  \log \Vert \rho_A\Vert  \leq & \HmaxB & \leq 2\log\tr\sqrt{\rho_A}.   \label{cor,eq2:upper bound for max-entropy}
\end{eqnarray}
Hence, $\HminB$ and $\HmaxB$ are finite if $\sqrt{\rho_A}$ is trace-class.
\end{p}


\section{\label{chapter:operational interpretation}Operational Interpretations of Min- and Max-Entropy}

Min- and max-entropy can be regarded as answers to operational questions, i.e., they quantify the optimal solution to certain information-theoretic tasks.
Max-entropy $\HmaxB$ answers the question of how distinguishable $\roo$ is from states that are maximally mixed on A, while uncorrelated with B  \cite{OperationalMeaning} (see also Definition~\ref{def:decoupling accuracy} below).
This is a useful concept, e.g., in quantum key distribution, where one ideally would have a maximally random key uncorrelated with the eavesdropper's state. Thus, the above distinguishability quantifies how well this is achieved.
Min-entropy $\HminB$ is related to the question of how close one can bring the state $\roo$ to a maximally entangled state on the bipartite system AB, allowing only local quantum operations on the B system \cite{OperationalMeaning}.
In the special case that A is classical (i.e., we have a classical-quantum state, see Eq.~(\ref{eq:cq state}) below) one finds that $\HminB$ is related to the guessing probability, i.e., our best chance to correctly guess the value of the classical system A, given the quantum system B.
In the following sections we show that these results can be generalized to the case that $H_{B}$ is  infinite-dimensional. These generalizations are for instance crucial in cryptographic settings, where there is a priori no reason to expect an eavesdropper to be limited to a finite-dimensional Hilbert space, while it is reasonable to assume the key to be finite. The operational interpretations of the min- and max-entropy exhibit a direct dependence on the dimension of the A system, which is why a naive generalization to an infinite-dimensional A appears challenging, and will not be considered here.


\subsection{ \label{section:oper interp of max entropy} Max-entropy as decoupling accuracy}

To define decoupling accuracy we use fidelity $F(\rho,\sigma) := \Vert \sqrt{\rho}\sqrt{\sigma}\Vert_1$ as a distance measure between states.
\begin{de}\label{def:decoupling accuracy}
For a finite-dimensional Hilbert space $\ha$ and an arbitrary separable Hilbert space $\hb$, we define the decoupling accuracy of $\rho_{AB} \in \tau_1^+(\hahb)$ w.r.t. the system B as
\begin{equation} \label{def,eq1:decoupling accuracy}
d(\roo|B) := \sup_{\sii \in \mc{S}(\hb)} d_A F( \roo , \tau_A \otimes \sii )^2 .
\end{equation}
Here, $d_A$ is the dimension of $\ha$, and $\tau_A := d_A^{-1} \id_A$ is the maximally mixed state on A.
\end{de}

Note that in infinite-dimensional Hilbert spaces there is no trace class operator which can be regarded as a generalization of the maximally mixed state in finite dimensions.  We must thus require system A to be finite-dimensional in order to keep the decoupling accuracy well-defined.
In \cite{OperationalMeaning}, Proposition~\ref{prop:operational interp of max entropy} was proved in the case where $H_B$ is assumed to be finite-dimensional. Below we use Proposition~\ref{p:reduction of Hmin to finite dim} to extend the assertion to the infinite-dimensional case.
\begin{p}\label{prop:operational interp of max entropy}
Let $\ha$ be a finite-dimensional and $\hb$ a separable Hilbert space. It follows that
\begin{equation} \label{prop,eq1:operational interp of max entropy}
 d(\roo|B) = 2^{H_{\mathrm{max}}(\roo|B)},
\end{equation}
for each $\roo \in \tau_1^+(\hahb)$.
\end{p}
In the following we will need to consider physical operations (channels) on states, i.e., trace preserving completely positive maps \cite{Kraus}. By $\TPCPM(\ha,\hb)$ we denote the set of all trace preserving completely positive maps $\mc{E}: \tau_1(H_A) \rightarrow \tau_1(H_B)$. Let $\mathcal{I}$ denote the identity map.

\begin{proof}
Let us take projected states $\roo^k$ relative to a generator of the form $(\id_A,P^B_k)$ (this is a proper generator since $\dim H_A<\infty$). Denote the space onto which $P^B_k$ projects by $U_k^B$ and set  $P_k := \id_A \otimes P_k^B$.
The finite-dimensional version of Proposition~\ref{prop:operational interp of max entropy} together with Proposition~\ref{p:reduction of Hmin to finite dim} yield $d(\roo^k|B_k)= 2^{H_{\mathrm{max}}(\roo^k|B_k)}\rightarrow 2^{H_{\mathrm{max}}(\roo|B)}$, as $k\rightarrow\infty$.

In order to prove $d(\roo|B) \leq 2^{H_{\mathrm{max}}(\roo|B)}$ we construct a suitable TPCPM and use the fact that the fidelity can only increase under its action \cite{NielsenChuang}.
For each $k \in \mathbb{N}$ choose a normalized state $\vert \theta_k \rangle \in \hahb$ such that $P_k\vert \theta_k \rangle =0$. We define a channel $\mc{E}_k \in \TPCPM(\hahb,\hahb)$ as
$\mc{E}_k(\eta) := P_k \eta P_k + q_k(\eta) \vert \theta_k \rangle \langle \theta_k \vert$,
with $q_k(\eta) :=\tr[ \eta (\id-P_k)]$. Then, for all $\sii \in \mc{S}(\hb)$ we find
\begin{align*}
F(\roo, \tau_{A} \otimes \sii) & \leq F\big( \mc{E}_k(\roo) , \mc{E}_k(\tau_{A} \otimes \sii) \big)\\
& = \big\Vert  \sqrt{\roo^k}\sqrt{\tau_{A} \otimes \sii^k} + \sqrt{q_k(\roo)q_k(\tau_{A} \otimes \sii)} \vert \theta_k \rangle \langle \theta_k \vert \, \big\Vert_1  \\
 & \leq \big\Vert \sqrt{\roo^k}\sqrt{\tau_{A} \otimes \sii^k} \big\Vert_1 + \sqrt{q_k(\roo)} = F(\roo^k , \tau_{A} \otimes \sii^k ) + \sqrt{q_k(\roo)},
\end{align*}
where $\sii^k := P_k^B \sii P_k^B$. The second line is due to the fact that $\vert \theta_k \rangle $ is  orthogonal to the support of both  $\roo^k$ and $\tau_{A} \otimes \sii^k$. The last line follows by the triangle inequality and $q_k(\tau_{A} \otimes \sii) \leq 1$. By taking the supremum over all $\sii \in \mc{S}(\hb)$ we obtain
\begin{equation*}
 \sqrt{d(\roo|B)} \leq \sqrt{d(\roo^k|B_k)} +  \sqrt{d_A \tr[ \roo (\id-P_k)]} \rightarrow 2^{\frac{1}{2} H_{\mathrm{max}}(\roo|B)},
\end{equation*}
as $k\rightarrow\infty$.
It remains to show $d(\roo|B) \geq 2^{H_{\mathrm{max}}(\roo|B)}$. For this purpose we use that the fidelity can be reformulated as
\begin{equation}\label{eq:Uhlmann}
F(\rho, \sigma) = \sup_{\vert \phi \rangle}  F(\vert \psi\rangle, \vert \phi \rangle),
\end{equation}
where $\vert \psi\rangle$ is a purification of $\rho$, and the supremum is taken over all purifications $\vert \phi \rangle$ of $\sigma$ \cite{Uhlmann1}.
Let us fix an arbitrary $k\in \mathbb{N}$ and a $\sii \in \mc{S}(\hb)$. Assume $\vert \psi_{ABC} \rangle$ to be a purification of $\roo$, and note that $\vert \psi^k_{ABC} \rangle := \tilde{P}_k\vert \psi_{ABC} \rangle$, with $\tilde{P}_k=P_k \otimes \id_C$, is a purification of $\roo^k$.
Let $\vert \phi \rangle\in H_{A}\otimes H_{B}\otimes H_{C}$ be an arbitrary purification of  $\tau_{A} \otimes \sii$. According to  (\ref{eq:Uhlmann}) it follows that
\begin{align*}
F(\roo , \tau_{A} \otimes \sii) & \geq F(\vert \psi_{ABC} \rangle ,\vert \phi \rangle )= | \langle  \psi_{ABC} \vert \phi \rangle | \\
& = | \langle  \psi_{ABC} \vert \tilde{P}_k \vert \phi \rangle + \langle  \psi_{ABC} \vert \id - \tilde{P}_k \vert \phi \rangle | \\
& \geq  | \langle  \psi^k_{ABC} \vert \phi \rangle | -  \Vert (\id-\tilde{P}_k )|\psi_{ABC}\rangle\Vert,
\end{align*}
where the last line is obtained by the reverse triangle inequality and the Cauchy-Schwarz inequality. By taking the supremum over all the purifications $|\phi\rangle$ of $\tau_{A} \otimes \sii$ in the above inequality, Eq.~(\ref{eq:Uhlmann}) yields $F(\roo , \tau_{A} \otimes \sii)  \geq F(\roo^k , \tau_{A} \otimes \sii)  -  \Vert (\id-\tilde{P}_k )|\psi_{ABC}\rangle\Vert$.
As this holds for all $\sii \in \mc{S}(\hb)$ and all $k$, we obtain with the definition of the decoupling accuracy:
\begin{equation*}
 d(\roo|B) \geq \lim_{k \rightarrow \infty} \Big(\sqrt{d(\roo^k|B_k)} -  \sqrt{d_A} \ \Vert (\id-\tilde{P}_k )|\psi_{ABC}\rangle\Vert \Big)^2 = 2^{H_{\mathrm{max}}(\roo|B)}.
 \end{equation*}

\end{proof}


\subsection{\label{section:Min-entropy is maximum achievable quantum correlation} Min-entropy as maximum achievable quantum correlation}

Assume a bipartite quantum system consisting of a finite-dimensional A system and an arbitrary B system. We can then define a maximally entangled state between the A and B system as
\begin{equation}\label{max entangled state}
\vert \Psi_{AB} \rangle := \frac{1}{\sqrt{d_A}}\sum_{k=1}^{d_A} \vert a_k \rangle\vert b_k \rangle.
\end{equation}
Here, $d_A$ denotes the dimension of $\ha$, $\{a_k\}_{k=1}^{d_A}$ an arbitrary orthonormal basis of $H_A$ and $\{b_k\}_{k=1}^{d_A}$ an arbitrary orthonormal system in $H_B$, where we assume that $\dim(H_A) \leq \dim(H_B)$.

\begin{de}\label{def:quantum correlation}
For $\ha$ a finite-dimensional and $\hb$ a separable Hilbert space (with $\dim H_A \leq \dim H_B$), we define the quantum correlation of a state $\rho_{AB} \in \mc{S}(\hahb)$ relative to B as
\begin{equation}\label{def,eq1:quantum correlation}
q(\roo|B) := \sup_{\mc{E}} d_A F\big(( \mathcal{I}_A \otimes \mc{E})\roo , \mcor\big)^2 ,
\end{equation}
where the supremum is taken over all $\mc{E}$ in TPCPM($H_B$,$H_B$), and $\vert \Psi_{AB}\rangle$ is given by (\ref{max entangled state}).
\end{de}
Due to the invariance of the fidelity under unitaries \cite{NielsenChuang}, the definition of $q(\rho_{AB}|B)$ is independent of the choice of the maximally entangled state $\vert \Psi_{AB} \rangle$. The quantum correlation can be rewritten as
\begin{equation}\label{eq:quantum correlation equiv expression}
q(\roo|B) = \sup_{\mc{E}} d_A \langle \Psi_{AB} \vert (\mathcal{I}_A \otimes \mc{E})\roo \vert \Psi_{AB} \rangle.
\end{equation}
The min-entropy is directly linked to the quantum correlation as shown in \cite{OperationalMeaning} for the finite-dimensional case. We extend this result to a B system with a separable Hilbert space.
\begin{p}\label{prop:oper interpr of min entropy}
Let $H_A$ be a finite-dimensional and $H_B$ be a separable Hilbert space. It follows that
\begin{equation}
q(\roo|B) = 2^{-H_{\mathrm{min}}(\roo|B)}  ,
\end{equation}
for each $\roo \in \mc{S}(\hahb)$.
\end{p}

\begin{proof}
Let $\{\roo^k\}_{k\in\mathbb{N}}$ be the projected states of $\roo$ relative to a generator of the form $(\id_A,P_k^B)$, and set $P_k := \id_A \otimes P_k^B$.
Let us denote the projection space of $P_k^B$ by $U_k^B$ and assume that $\vert b_l\rangle \in U_k^B$, $l=1,...,d_A$, for all $k$, with $\vert b_l\rangle$ as in equation (\ref{max entangled state}).
By the already proved finite-dimensional version of Proposition~\ref{prop:oper interpr of min entropy} and Proposition~\ref{p:reduction of Hmin to finite dim}, we obtain  $q(\roo^k|B_k) = \Lambda(\roo^k|B_k) \rightarrow \Lambda(\roo|B)$.

We begin to prove $\Lambda(\roo|B) \leq q(\roo|B)$. Fix $k$ and choose $\mc{E}_k \in \TPCPM(U_k^B,U_k^B)$ such that
$ q(\roo^k|B_k) =  d_A \langle \Psi_{AB} \vert (\mathcal{I}_A \otimes \mc{E}_k)\roo^k \vert \Psi_{AB} \rangle$.
Define $\tilde{\mc{E}}_k(\rho)  = \mathcal{E}_k(P_k\rho P_k) + (\id_B-P_k^B)\rho (\id_B -P_k^B)$, which is a valid quantum operation in $\TPCPM(\hb,\hb)$.
As $\tilde{\mc{E}}_k$ is just one possible TPCPM, it follows that
\begin{equation*}
q(\roo|B) \geq   d_A \langle \Psi_{AB} \vert (\mathcal{I}_A \otimes \tilde{\mc{E}}_k)\roo \vert \Psi_{AB} \rangle  \geq   q(\roo^k|B_k).
\end{equation*}
We thus find $q(\roo|B) \geq \lim_{k \rightarrow \infty} q(\roo^k|B_k) = \Lambda(\roo|B)$.

We next prove $\Lambda(\roo|B) \geq q(\roo|B)$. Let $\mc{E}$ be an arbitrary $\TPCPM(\hb,\hb)$. As a special instance of Stinespring dilations we know that there exists an ancilla $H_R$ together with an unitary $U_{BR} \in \mc{L}(\hb \otimes H_R)$ and a state $\vert \theta_R \rangle \in H_R$, such that
$\mc{E}(\sigma_B) = \tr_R[U_{BR} ( \sigma_B \otimes \vert \theta_R \rangle \langle \theta_R \vert )U^{\dagger}_{BR}]$ \cite{Kraus}. With $\vert \psi_{ABC} \rangle$ a purification of $\roo$, it follows according to (\ref{eq:Uhlmann}) that
\begin{align*}
 F \big(( \mathcal{I}_{A} \otimes \mc{E})\roo , \mcor\big)
& = \sup_{\eta_{CR}} F\big((\id \otimes U_{BR}) \vert \psi_{ABC} \rangle\vert\theta_R \rangle , \vert \Psi_{AB} \rangle\vert \eta_{CR} \rangle\big) \\
& \leq \sup_{\eta_{CR}} F\big(\rho_{AC} , \tau_A \otimes \tr_R(\vert \eta_{CR} \rangle \langle \eta_{CR} \vert)\big),
\end{align*}
where the last inequality is due to the monotonicity of fidelity under the partial trace and $\tau_A = d_A^{-1} \id_A = \tr_B(\mcor)$. The optimization over all pure states $\eta_{CR}$ can be replaced by the optimization over all density operators on $H_C$. Then, with Proposition~\ref{prop:operational interp of max entropy} it follows that
\begin{align*}
d_A F(( \mathcal{I}_{A} \otimes \mc{E})\roo , \mcor)^2  & \leq \sup_{\sigma_C} d_A F( \rho_{AC}, \tau_A \otimes \sigma_C)^2 = 2^{H_{\mathrm{max}}(\rho_{AC}|C)} \\
&  = 2^{-H_{\mathrm{min}}(\roo|B)} = \Lambda(\roo|B).
\end{align*}
Since this holds for all $\mc{E} \in \TPCPM(\hb,\hb)$, we obtain $q(\roo|B) \leq \Lambda(\roo|B)$.
\end{proof}


The quantum correlation and its relation to min-entropy applied to classical quantum states connects the min-entropy with the optimal guessing probability.
Imagine a source that produces the quantum states $\roBx \in \mc{S}(H_B)$ at random, according to the probability distribution $P_X(x)$. The average output is characterized by the classical-quantum state,
\begin{equation}\label{eq:cq state}
 \rho_{XB} = \sum_{x \in X} P_X(x) \vert x \rangle \langle x \vert \otimes \roBx,
\end{equation}
where $X$ denotes the (finite) alphabet of the classical system describing the source and $\{ \vert x \rangle \}_{x\in X}$ is an orthonormal basis spanning $H_X$. We define the guessing probability $g(\rho_{XB}|B)$ as the probability to correctly guess $x$, permitting an optimal measurement strategy on subsystem $B$. Formally, this can be expressed as
\begin{equation}\label{eq:def of guessing prob}
g(\rho_{XB}|B) := \sup_{ \{M_x\}} \sum_{x\in X}P_X(x) \tr(\roBx M_x),
\end{equation}
where the supremum is taken over all positive operator valued measures (POVM) on $H_B$. By POVM on $H_B$ we intend a set $\{M_x\}_{x \in X}$ of positive operators which sum up to the identity. For finite-dimensional $H_B$ it is known \cite{OperationalMeaning} that the guessing probability is linked to the min-entropy by
\begin{equation}\label{eq:guessin prob finite dim}
g(\rho_{XB}|B) = 2^{-H_{\mathrm{min}}(\rho_{XB}|B)}.
\end{equation}
We will now use Proposition~\ref{prop:oper interpr of min entropy} to show that Eq.~(\ref{eq:guessin prob finite dim}) also holds for separable $H_B$.

Let $\rho_{XB}$ be a state as defined in Eq.~(\ref{eq:cq state}), and construct the state $|\Psi_{XB}\rangle := |X|^{-1/2}\sum_{x\in X}|x \rangle|x_B\rangle$, where $\{|x_B\rangle\}_{x\in X}$ is an arbitrary orthonormal set in $H_B$. We now define
$Q(\rho_{XB},\mathcal{E}) :=  d_A \langle \Psi_{XB} \vert (\mathcal{I}_X \otimes \mc{E})\rho_{XB} \vert \Psi_{XB} \rangle$ (cf.  Eq.~(\ref{eq:quantum correlation equiv expression}))
 and $G(\rho_{XB}, \{M_x\}) := \sum_{x\in X}P_X(x) \tr(\roBx M_x)$ (cf. Eq.~(\ref{eq:def of guessing prob})).
Then,
\begin{equation}
\label{nadflv}
Q(\rho_{XB},\mathcal{E}) = \sum_{x\in X}P_X(x)\tr[\mathcal{E}^{*}(|x_B\rangle\langle x_B|)\rho_B^{x}],
\end{equation}
where $\mathcal{E}^{*}$ denotes the adjoint operation of $\mathcal{E}$. Let $\{M_x\}$ be an arbitrary $|X|$-element POVM on $H_B$. One can see that the TPCPM $\mathcal{E}(\rho) := \sum_{x\in X} \tr(M_x \rho) |x_B\rangle\langle x_B|$ satisfies $\mathcal{E}^{*}(|x_B\rangle\langle x_B|) = M_x$. Thus, by Eq.~(\ref{nadflv}), we find $Q(\rho_{XB},\mathcal{E}) = G(\rho_{XB}, \{M_x\})$. Since the POVM was arbitrary, it follows that $q(\rho_{XB}|B)\geq  g(\rho_{XB}|B)$.

Next, let  $\mathcal{E}$ be an arbitrary TPCPM on $H_B$. Define $P := \sum_{x\in X} |x_B\rangle\langle x_B|$ and
\begin{equation*}
M_{x} := \mathcal{E}^{*}(|x_B\rangle\langle x_B|) + \frac{1}{|X|}\mathcal{E}^{*}(\id_B-P),\quad x\in X.
\end{equation*}
One can verify that $\{M_x\}$ is a POVM on $H_B$. By using Eq.~(\ref{nadflv}) we can see that $G(\rho_{XB}, \{M_x\}) \geq Q(\rho_{XB},\mathcal{E})$. This implies $g(\rho_{XB}|B)\geq  q(\rho_{XB}|B)$, and thus $g(\rho_{XB}|B) = q(\rho_{XB}|B)$.


\section{\label{chapter:smooth entropy} Smooth Min- and Max-Entropy}

The entropic quantities that usually appear in operational settings are the smooth min- and max-entropies \cite{ChannelCodingMinEntropy,ChannelCodingMaxEntropy,OperationalMeaning}. They result from the non-smoothed versions by an optimization procedure over states close to the original state. The closeness is defined by an appropriate metric on the state space, and a smoothing parameter specifies the maximal distance to the original state. The choice of metric has varied in the literature, but here we follow  \cite{OnTheSmoothing}.

By $\mc{S}_{\leq}(H)$ we denote  the set of positive trace class operators with trace norm smaller than or equal to 1. We define the generalized fidelity on $\mc{S}_{\leq}(H)$ by $\bar{F}(\rho,\sigma):= \Vert\sqrt{\rho}\sqrt{\sigma}\Vert_{1} + \sqrt{(1-\tr\rho)(1-\tr\sigma)}$, which induces a metric on $\mc{S}_{\leq}(H)$ via
\begin{equation} \label{purified distance}
P(\rho,\sigma) := \sqrt{1 - \bar{F}(\rho,\sigma)^2},
\end{equation}
referred to as the purified distance.

\begin{de}\label{def:smooth min-/max-entropy}
For $\epsilon >0$, we define the $\epsilon$-smooth min- and max-entropy of $\roo \in \mc{S}_{\leq}(\hahb)$ conditioned on $B$ as
\begin{equation}\label{def,eq1:smooth min-/max-entropy}
\sHmin := \ssup \HminBt,
\end{equation}
\begin{equation}\label{def,eq2:smooth min-/max-entropy}
\sHmax := \sinf \HmaxBt,
\end{equation}
where the smoothing set $\mc{B}^{\epsilon}(\roo)$ is defined with respect to the purified distance
\begin{equation}\label{def,eq3:smooth min-/max-entropy}
\mc{B}^{\epsilon}(\roo) := \{ \root \in \mc{S}_{\leq}(H_A\otimes H_B)|   P(\roo,\root) \leq \epsilon \}.
\end{equation}
\end{de}

Closely related to this particular choice of smoothing set is the invariance of the smooth entropies under (partial) isometries acting locally on each of the subsystems. This can be used to show the duality relation of the smooth entropies, namely, for all states $\roo$ on $\hahb$ it follows that
\begin{equation}\label{eq:duality of smooth entropies}
\sHmin = - H_{\mathrm{max}}^{\epsilon}(\rho_{AC}|C),
\end{equation}
where $\rho_{ABC}$ is an arbitrary purification of $\roo$ on an ancilla $H_C$.
A proof for the finite-dimensional case can be found in \cite{OnTheSmoothing}, which allows a straightforward modification to infinite dimensions.

A useful property of the smooth entropies is the data processing inequality.
\begin{p}\label{p:strong subadditivity for smooth entropies}
Let be $ \rho_{ABC} \in \mc{S}_{\leq}(\hahb \otimes H_C)$, then it follows that
\begin{eqnarray*}
 H_{\mathrm{min}}^{\epsilon}(\rho_{ABC}|BC) &\leq & \sHmin,\\
 H_{\mathrm{max}}^{\epsilon}(\rho_{ABC}|BC) &\leq & \sHmax.
\end{eqnarray*}
\end{p}

\begin{proof}
Using the data processing inequality for the min-entropy, Eq.~(\ref{p,eq2:strong subadd of Hmin}), we obtain
\begin{eqnarray*}
H_{\mathrm{min}}^{\epsilon}(\rho_{ABC}|BC)=\sup_{\tilde{\rho}_{ABC} \in \mc{B}^{\epsilon}(\rho_{ABC})} H_{\mathrm{min}}(\rot_{ABC}|BC) \leq
\sup_{\tilde{\rho}_{ABC} \in \mc{B}^{\epsilon}(\rho_{ABC})}  H_{\mathrm{min}}(\tr_C \rot_{ABC}|B).
\end{eqnarray*}
Thus, it is sufficient to show that $\tr_C( \ball(\rho_{ABC})) \subseteq \ball(\roo)$. But this follows directly from the fact that the purified distance does not increase under partial trace \cite{OnTheSmoothing}, i.e., $P(\rho_{ABC},\rot_{ABC})\geq P(\rho_{AB},\rot_{AB})$.

The data processing inequality of the smooth max-entropy follows from the duality (\ref{eq:duality of smooth entropies}),
\begin{equation*}
 H_{\mathrm{max}}^{\epsilon}(\rho_{ABC}|BC) = -H_{\mathrm{min}}^{\epsilon}(\rho_{AD}|D) \leq- H_{\mathrm{min}}^{\epsilon}(\rho_{ACD}|CD) = \sHmax,
\end{equation*}
where $\rho_{ABCD}$ is a purification of $\rho_{ABC}$.
\end{proof}


\section{\label{AEP} An Infinite-Dimensional Quantum Asymptotic Equipartition Property}

In the finite-dimensional case the quantum asymptotic equipartition property (AEP) says that the conditional von Neumann entropy can be regained as an asymptotic quantity from the conditional smooth min- and max-entropy \cite{RennerPhD,QuantumAEP}. (For a discussion on why the AEP can be formulated in terms of entropies, see \cite{Independent}.) More precisely, $\lim_{\epsilon \rightarrow 0} \lim_{ n \rightarrow \infty} \frac{1}{n} H_{\mathrm{min}}^{\epsilon}(\roo^{\otimes n}|B^n) =  H(\roo|B)$ and
$\lim_{\epsilon \rightarrow 0} \lim_{ n \rightarrow \infty} \frac{1}{n} H_{\mathrm{max}}^{\epsilon}(\roo^{\otimes n}|B^n) =  H(\roo|B)$.
For the infinite-dimensional case we derive an upper (lower) bound to the conditional von Neumann entropy in terms of the smooth min-(max-)entropy. We then use these bounds to prove the above limits in the case where $H_A$ is finite-dimensional.
To this end we need a well defined notion of conditional von Neumann entropy in the infinite-dimensional case. Here we use the definition introduced in \cite{Kuznetsova}, which in turn is based on an infinite-dimensional extension of the relative entropy \cite{Klein31,Lindblad73,Lindblad74,HolevoShirokov}.
For $\rho , \sigma \in \tau_1^{+}(H)$ the relative entropy can be defined as
\begin{equation}
\label{smdal}
H(\rho \Vert \sigma): = \sum_{jk} |\langle a_{j}|b_{k}\rangle|^{2}(a_j\log a_j - a_j\log b_k + b_k-a_j),
\end{equation}
where $\{|a_j\rangle\}_j$ is an arbitrary orthonormal eigenbasis of $\rho$ with corresponding eigenvalues $a_j$, and analogously for $\{|b_k\rangle\}_k$, $b_k$, and $\sigma$. The relative entropy is always positive, possibly $+\infty$, and equal to 0 if and only if $\rho=\sigma$ \cite{Lindblad73}.
For states $\roo$ with $H(\rho_A) < +\infty$, the conditional von Neumann entropy can be defined as \cite{Kuznetsova}
\begin{equation}
H(\roo|B) := H(\rho_A) -H(\rho_{AB}\Vert \rho_A\otimes \rho_B).
\end{equation}
For many applications it appears reasonable to assume $H(\rho_A)$ to be finite, e.g., in cryptographic settings it would correspond to restricting the states of the `legitimate' users.

Similarly as for the min- and max-entropy, the conditional von Neumann entropy can be approximated by projected states \cite{Kuznetsova}, i.e., for $\roo \in \mc{S}(H_A \otimes H_B)$ satisfying $H(\rho_A)<\infty$ with corresponding normalized projected states $\rooh^k$ it follows that
\begin{equation}\label{eq: limit von neumann entropy}
\lim_{k \rightarrow \infty } H(\rooh^k|B) = H(\roo|B).
\end{equation}
In the finite-dimensional case it has been shown \cite{QuantumAEP} that the min-, max- and, von Neumann entropy can be ordered as
\begin{equation}
\label{ordering}
\HminB \leq H(\roo| B) \leq \HmaxB.
\end{equation}
A direct application of Proposition~\ref{p:reduction of Hmin to finite dim} and (\ref{eq: limit von neumann entropy}) shows that this remains true in the infinite-dimensional case, if $H(\rho_A)<\infty$.
Note, however, that the ordering between min- and max-entropy (\ref{orderingHminHmax}) does not hold for their smoothed versions.

\begin{p} \label{prop:AEP lower bound}
Let $\roo  \in \mc{S}(\hahb)$ be such that $H(\rho_A)<\infty$. For any  $\epsilon > 0$ it follows that
\begin{equation}
\label{lowerbound}
\frac{1}{n} H_{\mathrm{min}}^{\epsilon}(\rho_{AB}^{\otimes n} | B^n) \geq H(\roo|B) - \frac{1}{\sqrt{n}}4\log(\eta) \sqrt{\log\frac{2}{\epsilon^2}},
\end{equation}
\begin{equation}
\label{upperbound}
\frac{1}{n} H_{\mathrm{max}}^{\epsilon}(\rho_{AB}^{\otimes n} | B^n) \leq H(\rho_{AB}|B) + \frac{1}{\sqrt{n}}4\log(\eta) \sqrt{\log\frac{2}{\epsilon^2}}.
\end{equation}
for $n \geq (8/5)\log(2/\epsilon^2)$, and
 $\eta = 2^{-\frac{1}{2}\HminB} + 2^{\frac{1}{2}\HmaxB} +1$.
\end{p}

Note that it is not clear under what conditions the limits $n\rightarrow \infty$, $\epsilon \rightarrow 0$  exist for the left hand side of equations (\ref{lowerbound}) and (\ref{upperbound}). If they do, Proposition~\ref{prop:AEP lower bound} implies
$\lim_{\epsilon\rightarrow 0}\lim_{n\rightarrow\infty}\frac{1}{n}
H_{\mathrm{min}}^{\epsilon}(\rho_{AB}^{\otimes n} | B^n) \geq H(\roo|B)$ and
$\lim_{\epsilon\rightarrow 0}\lim_{n\rightarrow\infty}\frac{1}{n}
H_{\mathrm{max}}^{\epsilon}(\rho_{AB}^{\otimes n} | B^n) \leq H(\roo|B)$.
For the case of a finite-dimensional $H_A$ we show that these inequalities can be replaced with equalities  (Corollary~\ref{cor:AEP}).

It should be noted that in the classical case  a lower bound on the min-entropy and an upper bound on the max-entropy, analogous to Eqs.~(\ref{lowerbound}) and (\ref {upperbound}), correspond \cite{Independent} to the AEP in classical probability theory \cite{CoverThomas}. Since in the finite-dimensional quantum case, the step from Proposition~\ref{prop:AEP lower bound} to Corollary~\ref{cor:AEP} is directly obtained \cite{QuantumAEP} via Fannes' inequality \cite{Fannes}, the limits in Corollary~\ref{cor:AEP} are usually referred to as `the quantum AEP' \cite{QuantumAEP}.
In the infinite-dimensional case  the relation between Proposition~\ref{prop:AEP lower bound} and Corollary~\ref{cor:AEP} appears  less straightforward, and it is thus not entirely clear what should be regarded as constituting `the quantum AEP'.
We will not pursue this question here, but merely note that it is the inequalities in Proposition~\ref{prop:AEP lower bound}, rather than the limits in Corollary~\ref{cor:AEP}, that are the most relevant for applications \cite{RennerPhD}. However, for the sake of simplicity we continue to refer to Corollary~\ref{cor:AEP} as a quantum AEP.

We prove Proposition~\ref{prop:AEP lower bound} after the following lemma.
\begin{lem} \label{cor: eps dep smooth min-entropy}
Let $\roo \in \mc{S}(\hahb)$ and let $\{\rooh^k\}_{k=1}^{\infty}$ be a sequence of normalized projected states. For any fixed $1 >t>0$, there exists a $k_0 \in \mathbb{N}$ such that
\begin{equation}
 H_{\mathrm{min}}^{\epsilon}(\roo|B) \geq H_{\mathrm{min}}^{t\epsilon}(\rooh^k|B),\quad \forall k\geq k_0.
\end{equation}
\end{lem}

\begin{proof}
In the following let $t\in (0,1)$ be fixed. According to the definition of the smooth min-entropy in Eq.~(\ref{def,eq1:smooth min-/max-entropy}),
it is enough to show that $\mc{B}^{t\epsilon}(\rooh^k) \subseteq \ball(\roo)$ for all $k \geq k_0$.
Note that the purified distance is compatible with trace norm convergence, i.e.,
 $\Vert \roo-\rooh^k\Vert_1\rightarrow 0$ implies that $P(\rooh^k,\roo) \rightarrow 0$.
Hence, there exists a $k_0$ such that $P(\rooh^k,\roo) < (1-t)\epsilon$ for all $k \geq k_0$. For $k \geq k_0$ and $\root \in \mc{B}^{t\epsilon}(\rooh^k)$ we thus find $P(\root,\roo) \leq P(\root, \rooh^k) + P(\rooh^k,\roo) < \epsilon$, such that $\root \in \ball(\roo)$.
\end{proof}

\begin{proof}\emph{(Proposition~\ref{prop:AEP lower bound})}
Let   $(P_k^A, P_k^B)$ be a generator of projected states. The pair of n-fold tensor products of the projections, $\big((P_k^A)^{\otimes n},(P_k^B)^{\otimes n}\big)$, is also a generator of projected states.
If we now fix $1>t>0$ and $n \in \mathbb{N}$, it follows by Lemma~\ref{cor: eps dep smooth min-entropy} that we can find a $k_0\in \mathbb{N}$  such that $H_{\mathrm{min}}^{\epsilon}(\roo^{\otimes n}|B^n) \geq H_{\mathrm{min}}^{t\epsilon}((\rooh^k)^{\otimes n}|B^n)$ for every $k\geq k_0$.
Since Eq.~(\ref{lowerbound}) is valid for the finite-dimensional case  \cite{QuantumAEP}, we can apply it to $H_{\mathrm{min}}^{t\epsilon}((\rooh^k)^{\otimes n}|B^n)$ to obtain
\begin{equation*}
\frac{1}{n}H_{\mathrm{min}}^{t\epsilon}((\rooh^k)^{\otimes n}|B^n) \geq
H(\rooh^k|B) - \frac{1}{\sqrt{n}}4\log(\eta_k) \sqrt{\log\frac{2}{(t\epsilon)^2}}
\end{equation*}
for any $n \geq (8/5)\log(2/(t\epsilon)^2)$, and $\eta_k = 2^{-\frac{1}{2}H_{\mathrm{min}}(\rooh^k|B) } + 2^{\frac{1}{2}H_{\mathrm{max}}(\rooh^k|B)} +1$. Hence
\begin{equation}\label{pf,lem:AEP lower bound,eq4}
\frac{1}{n} H_{\mathrm{min}}^{\epsilon}(\roo^{\otimes n}|B^n)   \geq H(\rooh^k|B) - \frac{1}{\sqrt{n}}4\log(\eta_k) \sqrt{\log\frac{2}{(t\epsilon)^2}},
\end{equation}
for all $k \geq k_0$.
Since the left hand side of Eq.~(\ref{pf,lem:AEP lower bound,eq4}) is independent of $k$ we can use
(\ref{eq: limit von neumann entropy}) and Proposition~\ref{p:reduction of Hmin to finite dim}, to find
\begin{align*}
\frac{1}{n} H_{\mathrm{min}}^{\epsilon}(\roo^{\otimes n}|B^n)
& \geq \lim_{k\rightarrow \infty} \Big\{H(\rooh^k|B) - \frac{1}{\sqrt{n}}4\log(\eta_k) \sqrt{\log\frac{2}{(t\epsilon)^2}} \Big\} \\
& = H(\roo|B) - \frac{1}{\sqrt{n}}4\log(\eta) \sqrt{\log\frac{2}{(t\epsilon)^2}}.
\end{align*}
We finally take the limit $t\rightarrow 1$ in the above inequality, as well as in the condition $n \geq (8/5)\log(2/(t\epsilon)^2)$ to obtain the first part of the proposition.

For the second part we use the duality of the conditional von Neumann entropy, i.e., $H(\rho_{AB}|B)= -H(\rho_{AC}|C)$ for a purification $\rho_{ABC}$ \cite{Kuznetsova}. This, together with the duality relation for smooth min- and max-entropy (\ref{eq:duality of smooth entropies}) leads directly to (\ref{upperbound}).
\end{proof}

\begin{cor} \label{cor:AEP} Let $H_A$ be a finite-dimensional and $H_B$ a separable Hilbert space. For all $\roo  \in \mc{S}(\hahb)$ it follows that
\begin{equation}\label{cor:AEP min-entropy}
\lim_{\epsilon\rightarrow 0} \lim_{n \rightarrow \infty} \frac{1}{n} H_{\mathrm{min}}^{\epsilon}(\rho_{AB}^{\otimes n} | B^n) = H(\roo|B),
\end{equation}
\begin{equation}\label{cor:AEP max-entropy}
\lim_{\epsilon\rightarrow 0} \lim_{n \rightarrow \infty} \frac{1}{n} H_{\mathrm{max}}^{\epsilon}(\rho_{AB}^{\otimes n} | B^n) = H(\roo|B).
\end{equation}
\end{cor}

\begin{proof} Let $\epsilon >0$ be sufficiently small, and let $(\id_A,P_k^B)$ be a generator of projected states $\roo^k$, with corresponding normalized projected states $\hat{\rho}_{AB}^k$. Let $\sigma_{AB}\in \ball(\roo)$, with projected states $\sigma^k_{AB}$, and normalized projected states $\hat{\sigma}^k_{AB}$.  By $H_{\mathrm{min}}(\sigma_{AB}^k|B)=H_{\mathrm{min}}(\sih^k_{AB}|B)+\log\tr\sigma^k_{AB}$ and (\ref{ordering}) we find $H_{\mathrm{min}}(\sigma^k_{AB}|B_k)
 \leq H(\sih^k_{AB}|B)$, where $\sih^k_{AB}=(\tr\sigma^k_{AB})^{-1}\sigma^k_{AB}$. Since  $H(\sih^k_{AB}|B_k)$ is finite-dimensional we can use Fannes' inequality \cite{Fannes} to obtain (for $k$ sufficiently large)
$H(\sih^k_{AB}|B_k) \leq H(\rooh^k|B_k) + 4 \Delta_k \log d_A + 4 H_{\mathrm{bin}}(\Delta_k)$,
with $d_A=\dim(H_A)$, $\Delta_k = \Vert \rooh^k - \sih^k_{AB} \Vert_1$, and $H_{\mathrm{bin}}(t)= -t\log t - (1-t)\log(1-t)$. Due to the general relation $\Vert \rho-\sigma\Vert_1 \leq 2 P(\rho,\sigma)$ (see Lemma 6 in \cite{OnTheSmoothing}), we have $\Vert \roo-\sigma_{AB}\Vert_1 \leq 2\epsilon$ for all $\sigma_{AB}\in\ball(\roo)$,
which yields $\lim_{k\rightarrow \infty} \Delta_k = \Vert \roo -\sih_{AB} \Vert_1 \leq 4\epsilon$, where $\sih_{AB} = \sigma_{AB}/\tr(\sigma_{AB})$.  Combined with (\ref{eq: limit von neumann entropy}) this leads to $H^{\epsilon}_{\mathrm{min}}(\rho_{AB} | B) = \ssupp\lim_{k\rightarrow\infty}H_{\mathrm{min}}(\sigma^k_{AB} | B) \leq H(\roo|B) + 16\epsilon \log d_A + 4 H_{\mathrm{bin}}(4\epsilon)$. Applied to an n-fold tensor product this gives
\begin{equation}
\label{MinEntrUpperbound}
\frac{1}{n}H_{\mathrm{min}}^{\epsilon}(\roo^{\otimes n} | B^n) \leq H(\roo|B) + 16\epsilon \log d_A + \frac{4}{n} H_{\mathrm{bin}}(4\epsilon).
\end{equation}
 Equation (\ref{cor:AEP min-entropy}) follows by combining (\ref{MinEntrUpperbound}) with the lower bound in (\ref{lowerbound}), taking the limits $n\rightarrow \infty$ and $\epsilon\rightarrow 0$.
 Equation (\ref{cor:AEP max-entropy}) follows directly by the duality of the conditional von Neumann entropy \cite{Kuznetsova} together with the duality of the smooth min- and max-entropy (\ref{eq:duality of smooth entropies}).
\end{proof}


\section{\label{concl} Conclusion and Outlook}

We have extended the min- and max-entropies to separable Hilbert spaces, and shown that  properties and operational interpretations, known from the finite-dimensional case, remain valid in the infinite-dimensional setting. These extensions are facilitated by the finding (Proposition~\ref{p:reduction of Hmin to finite dim}) that the infinite-dimensional min- and max-entropies can be expressed in terms of convergent sequences of finite-dimensional entropies.
We bound the smooth min- and max-entropies of iid states (Proposition \ref{prop:AEP lower bound}) in terms of an infinite-dimensional generalization of the conditional von Neumann entropy $H(A|B)$, introduced in \cite{Kuznetsova}, which is defined when the von Neumann entropy of system $A$ is finite, $H(A)<\infty$. Under the additional assumption that the Hilbert space of system $A$ has finite dimension we furthermore prove that the smooth entropies of iid states converge to the conditional von Neumann entropy  (Corollary \ref{cor:AEP}), corresponding to a quantum asymptotic equipartition property (AEP). Whether these conditions can be relaxed is an open question. In the general case where $H(A)$ is not necessarily finite, this would however require a more general definition of the conditional von Neumann entropy than the one used here.

For information-theoretic purposes it appears reasonable to require extensions of the conditional von Neumann entropy to be compatible with the AEP, i.e., that the conditional von Neumann entropy can be regained from the smooth min- and max-entropy in the asymptotic iid limit. This enables generalizations of operational interpretations of the conditional von Neumann entropy. For example, in the finite-dimensional asymptotic case the conditional von Neumann entropy  characterizes the amount of entanglement needed for state merging \cite{StateMerging}, i.e., the transfer of a quantum state shared by two parties to only one of the parties. An infinite-dimensional generalization of one-shot state merging \cite{Single Shot State merging}, together with the AEP, could be used to extend this result to the infinite-dimensional case.

Some other immediate applications of this work are in continuous variable quantum key distribution, and in statistical mechanics, where it has recently been shown \cite{Oscar,Lidia} that the smooth min- and max-entropies play a role. Our techniques may also be employed to derive an infinite-dimensional generalization of the entropic uncertainty relation \cite{Uncertainty}. Such a generalization would be interesting partially because it could find applications in continuous variable quantum information processing, but also because it may bring this information-theoretic uncertainty relation into the same realm as the standard uncertainty relation.


\section{Acknowledgments}
We thank Roger Colbeck and Marco Tomamichel for helpful comments and discussions, and an anonymous referee for very valuable suggestions.
Fabian Furrer acknowledges support from the Graduiertenkolleg 1463 of the Leibniz University Hannover.
 We furthermore acknowledge support from the Swiss National Science Foundation (grant No. 200021-119868).

\begin{appendix}


\section{\label{section:Technical Lemmas}Technical Lemmas}

 In the following, each Hilbert space is assumed to be separable. Let us define the positive cone $\mc{L}^+(H) := \{T \in \mc{L}(H)| \ T \geq 0\}$ in $\mc{L}(H)$. The next two lemmas follow directly from the definition of positivity of an operator.
\begin{lem}\label{lem:T geq 0 implies STS geq 0}
If $T \in \mc{L}^+(H)$, then for each $S \in \mc{L}(H)$ it follows that $STS^{\dagger} \in \mc{L}^+(H)$.
\end{lem}

\begin{lem} \label{lem:pos weak op limit}
The positive cone $\mc{L}^+(H)$ is sequentially closed in the weak operator topology, i.e.,
for $\{T_k\}_{k\in\mathbb{N}} \subset \mc{L}^+(H)$ such that $T_k$ converge to $T \in \mc{L}(H)$ in the weak operator topology, it follows that $T\geq 0$.
\end{lem}

The following lemma is a special case of a theorem by Gr\"umm \cite{Grumm} (see also \cite{Simon}, pp. 25-29,  for similar results).
\begin{lem}
\label{lemgrumm}
Let $A_{k},A \in\mc{L}(H)$, such that $\sup_{k}\Vert A_{k}\Vert<+\infty$, and $A_{k}\rightarrow A$ in the strong operator topology, and let $T\in \tau_1(H)$. Then $\lim_{k\rightarrow\infty}\Vert A_{k}T -AT\Vert_{1}=0$ and $\lim_{k\rightarrow\infty}\Vert TA_{k} -TA\Vert_{1}=0$.
\end{lem}

\begin{cor}
\label{dfnbkl}
If $P_{k}$ is a sequence of projectors on $H$ that converges in the strong operator topology to the identity, and if $\rho \in \tau_1^+(H)$, then $\lim_{k\rightarrow\infty}\Vert P_{k}\rho P_{k}-\rho\Vert_{1}=0$.
\end{cor}
\begin{lem}
\label{nvdakj}
If sequences of projectors $P_k^A$ and $P_k^B$ on $H_A$ and $H_B$, respectively, converge in the strong operator topology to the identity, then  $P_k^A\otimes P_k^B$ converges in the strong operator topology to $\id_{AB}$.
\end{lem}

\begin{lem}\label{lem:weakstar Tk implies weak op of id otimes Tk}
Let $\{T_k\}_{k\in \mathbb{N}} \subset \tau_1(H_B)$ be a sequence that converges in the weak* topology to $T\in \tau_1(H_B)$. Then, the sequence $\id_{A} \otimes T_k$ in $\mc{L}(\hahb)$ converges to $\id_{A} \otimes T$ in the weak operator topology.
\end{lem}
\begin{proof}

For each $\psi \in \hahb$ we find that  $\langle \psi \vert \id \otimes T_k \vert \psi \rangle =  \tr( T_k K^{B}_{\psi})$, where $K^{B}_{\psi}= \tr_{A}|\psi\rangle\langle\psi|$ is the reduced operator. Since $K^{A}_{\psi}$ is trace class (and thus compact) the statement follows immediately.

\end{proof}


\section{\label{app:proofprop1}Proof of Proposition~\ref{p:reduction of Hmin to finite dim}}

In order to derive Proposition~\ref{p:reduction of Hmin to finite dim} we proceed as follows: In Section~\ref{Section:reduction to finite} we show that the min- and max-entropy of a projected state can be reduced to an entropy on a finite-dimensional space. In Section~\ref{sec:monotonicity} we show that the min- and max-entropies are monotonic over the sequences of projected states. Finally we prove the limits listed in Proposition~\ref{p:reduction of Hmin to finite dim}.  Note that in what follows we mostly make use of the quantities $\Lambda(\rho_{AB}|\sigma_{B})$ and $\Lambda(\rho_{AB}|B)$, as defined in Eqs.~(\ref{lambdadef1}) and (\ref{lambdadef2}), rather than the min- and max-entropies per se.

\subsection{\label{Section:reduction to finite} Reduction}

Here we show that the min- and max-entropy of a projected state can be considered as effectively finite-dimensional, in the sense that restricting the Hilbert space to the support of the projected states does not change the value of the entropies.

\begin{lem} \label{I:id-P}
Let $P_A$, $P_B$ be projectors onto closed subspaces $U_A \subseteq H_A$ and $U_B \subseteq H_B$, respectively, $\root \in \tau_{1}^{+}(H_A \otimes H_B)$, and $\siit \in \tau^+_1(H_B)$. \newline
i) If $(P_{A}\otimes \id_{B}) \root (P_{A}\otimes \id_{B}) = \root$ it follows that
      \begin{equation}
      \Lambda(\root|\siit)  = \inf \{ \lambda \in \mathbb{R} | \lambda P_A \otimes \siit \geq \root  \}.
  \end{equation}
ii) If $(\id_A\otimes P_{B})\root (\id_A\otimes P_{B}) = \root$ it follows that
    \begin{equation}\label{prop,eq1:Lambda statement for reduction of Hmin}
      \Lambda(\root|B) = \Lambda(\root|U_B),
    \end{equation}
where $\Lambda(\root |U_B)$ means that the infimum in Eq.~(\ref{lambdadef2}) is taken
only over the set $\tau^+_1(U^B)$.
\end{lem}
The proof is straightforward and left to the reader. In the particular case of projected states $\roo^k$ relative to a generator $(P_k^A,P_k^B)$, the evaluation of $\Lambda(\roo^k|\sii^k)$ and $\Lambda(\roo^k|B)$, where $\sii^k=P_k^B\sii P_k^B$, can be restricted to the finite-dimensional Hilbert space $U_k^A\otimes U_k^B$ given by the projection spaces of $P_k^A$ and $P_k^B$. Especially, we can conclude that the infima of Eqs.~(\ref{lambdadef1}) and (\ref{lambdadef2}), and consequently the infimum in (\ref{def,eq1:min/max-entropy}) and the supremum in (\ref{def,eq2:min/max-entropy}), are attained for projected states, since these are optimizations of continuous functions over compact sets.


\subsection{\label{sec:monotonicity}Monotonicity}
The next lemma considers the monotonic behaviour of the min- and max-entropies with respect to sequences of projected states.

\begin{lem} \label{l:mon incr}
For $\roo\in \mc{S}(\hahb)$, $ \sii \in \mc{S}(H_B)$, let $\{\roo^k\}_{k=1}^{\infty}$ and $\{\sii^k\}_{k=1}^{\infty}$ be projected states relative to a generator $(P_k^A,P_k^B)$.\newline
i) It follows that $\Lambda(\roo^k|\sii^k)$ and $\Lambda(\roo^k|B)$ are monotonically increasing in $k$, where the first sequence is bounded by $\Lambda(\roo|\sii)$ and the latter by $\Lambda(\roo|B)$.\newline
ii) For an arbitrary but fixed purification $\rho_{ABC}$ of $\rho_{AB}$ with purifying system $H_C$, let $\rho^k_{AC} = \tr_B\rho_{ABC}^k$ and $\rho_{ABC}^k = (P^A_k \otimes P_k^B \otimes \id_C)\rho_{ABC}(P^A_k \otimes P_k^B  \otimes \id_C)$. Then it follows that $\Lambda(\rho_{AC}^k|C)$ is monotonically increasing and bounded by $\Lambda(\rho_{AC}|C)$.
\end{lem}
Note that $\rho^k_{AC}$ as defined in the lemma is not a projected state in the sense of Definition~\ref{def:projected states}. Translated to min- and max-entropies, the lemma above says that $\Hmink$ and $\HminBk$ are monotonically increasing while $\HmaxBk$ is monotonically decreasing. But in general, the monotonicity does not hold for \emph{normalized} projected states.

\begin{proof}
Set $P_k:=P_k^A \otimes P_k^B$ and recall that $\Lambda(\roo^k|\sii^k)= \inf \{ \lambda \in \mathbb{R}| \ \lambda P_k^A \otimes \sii^k \geq \roo^k \}$ according to Lemma~\ref{I:id-P}.
To show the first part of i) note that for $k'\leq k$ the equations
\begin{equation*}
P_{k'}P_k (\lambda \id \otimes \sii - \roo ) P_{k'}P_k=P_{k'} (\lambda P_k^A \otimes \sii^k - \roo^k ) P_{k'} =  \lambda P_{k'}^A \otimes \sii^{k'} - \roo^{k'}
\end{equation*}
hold, which imply via Lemma~\ref{lem:T geq 0 implies STS geq 0} that $\Lambda(\roo^{k'}|\sii^{k'}) \leq \Lambda(\roo^k|\sii^k) \leq \Lambda(\roo|\sii)$.
For the second part, let $\siit \in \tau_1^+(H_B)$ be the optimal state such that $\Lambda(\roo^k|B)=\tr\siit$ and $P_k^A\otimes \siit \geq \roo^k$. But then we obtain that $P_{k'}^A \otimes P_{k'}^B\siit P_{k'}^B - \roo^{k'} \geq 0$ and therefore also $\Lambda(\roo^{k'}|B) \leq \Lambda(\roo^k|B)$. The upper bound follows in the same manner.

In order to show ii) we define the sets $ \mc{M}_k := \{\sitC \in \tau_1^+(H_C)|\ \id_A \otimes \sitC \geq \rho_{AC}^k\}$ such that $\Lambda(\rho_{AC}^k|C) = \inf_{\sitC \in \mc{M}_k}\tr\sitC$.
To conclude the monotonicity we show that $ \mc{M}_{k'} \supset \mc{M}_{k}$ for $k'\leq k$. If $\mc{M}_{k}= \emptyset$, the statement is trivial.
Assume $\tilde{\sigma}_{C}\in \mc{M}_k$. Using $P^B_{k'}\leq P^B_{k}$ we find
\begin{equation*}
\id_A\otimes \tilde{\sigma}_{C}\geq  P_{k}^A \tr_{B}(P_{k}^B \rho_{ABC}P^B_{k}) P_{k}^A \geq P_{k}^A \tr_{B}(P_{k'}^B \rho_{ABC}P^B_{k'}) P_{k}^A.
\end{equation*}
Together with Lemma~\ref{lem:T geq 0 implies STS geq 0}, this yields $P^A_{k'} \otimes \tilde{\sigma}_{C} \geq \rho^{k'}_{AC}$ and thus $\tilde{\sigma}_{C}\in \mc{M}_{k'}$. A similar argument provides the upper bound $\Lambda(\rho_{AC}^{k}|C) \leq  \Lambda(\rho_{AC}|C)$.
\end{proof}


\subsection{\label{sec:limits}Limits}

After the above discussion on general properties of the min- and max-entropies of projected states we are now prepared to prove Proposition~\ref{p:reduction of Hmin to finite dim}. For the sake of convenience we divide the proof into three lemmas.

\begin{lem} \label{l=sup}
For $\roo\in \mc{S}(\hahb)$ and $\sii \in \mc{S}(\hb)$, let $\{\roo^k\}_{k=1}^{\infty}$ be the projected states of $\roo$ relative to a generator $(P_k^A,P_k^B)$, and let $\sii^k:= P_k^B\sii P_k^B$. It follows that
\begin{equation}
\label{knbvx}
\Lambda(\roo|\sii) = \lim_{k \rightarrow \infty} \Lambda(\roo^k|\sii^k),
\end{equation}
and the infimum in Eq.~(\ref{lambdadef1}) is attained if  $\Lambda(\roo|\sii)$ is finite.
\end{lem}
\begin{proof}
That the infimum is attained follows directly from the definition. To show (\ref{knbvx}) we prove that $\Lambda(\roo|\sii)$ is lower semi-continuous in $(\roo,\sii)$ with respect to the product topology induced by the trace norm topology on each factor. Since this means that $\liminf_{k \rightarrow \infty} \Lambda(\roo^k|\sii^k) \geq \Lambda(\roo|\sii)$, the combination with Lemma~\ref{l:mon incr} results directly in (\ref{knbvx}). To show lower semi-continuity recall that it is equivalent to say that all lower level sets $\Lambda^{-1}((-\infty,t]) =\{(\roo,\sii) | \ \Lambda(\roo|\sii) \leq t \}$,
for $t\in \mathbb R$ have to be closed. But this follows by rewriting $\Lambda^{-1}((-\infty,t])$ as $\{(\roo,\sii) | \ t\id \otimes \sii \geq \roo \}$.
\end{proof}

\begin{lem}\label{l:Lambda statement for reduction of Hmin}
For $\roo\in \mc{S}(\hahb)$, let $\{\roo^k\}_{k=1}^{\infty}$ be the projected states of $\roo$ relative to a generator $(P_k^A,P_k^B)$. It follows that
\begin{equation}
\label{lakdsfn}
 \Lambda(\roo|B) = \lim_{k\rightarrow \infty} \Lambda(\roo^k|B),
\end{equation}
and the infimum in Eq.~(\ref{lambdadef2}) is attained if $\Lambda(\roo|B)$ is finite.
\end{lem}

\begin{proof}
Let $\mu_k := \Lambda(\roo^k|B) = \Lambda(\roo^k|B_k)$, where the last equality is due to Lemma~\ref{I:id-P}.
By Lemma~\ref{l:mon incr} this sequence is monotonically increasing, and we can thus define $\mu := \lim_{k \rightarrow \infty} \mu_k \in \mathbb{R} \cup \{+\infty\}$. In addition, Lemma~\ref{l:mon incr} also yields $\mu \leq \Lambda(\roo|B)$. Hence, the case $\lambda = +\infty$ is trivial, and it remains to show $\mu \geq \Lambda(\roo|B)$, for $\mu < \infty$.

For each $k\in \mathbb{N}$ let $\tilde{\sigma}_B^k$ be an optimal state such that $\Lambda(\roo^k|B)=\tr\siit^k$ and $\id\otimes\siit^k\geq \roo^k$. Note that due to positivity $\tr\siit^k=\Vert \siit^k\Vert_1
\leq \mu$, such that $\tilde{\sigma}_B^k$ is a bounded sequence in $\tau_1(H_B)$. Since the trace class operators $\tau_1(H_B)$ is the dual space of the compact operators $\mc K(H_B)$ \cite{SimonReeds}, we can apply Banach Alaoglu's theorem \cite{SimonReeds,HillPhillips} to find a subsequence $\{\tilde{\sigma}_B^k\}_{k\in\Gamma}$ with a weak* limit $\siit \in \tau_1(H_B)$, i.e., $\tr(K\siit^k) \rightarrow \tr(K\siit)$ $(k\in \Gamma)$ for all $K\in \mc K(H_B)$, such that $\Vert \siit\Vert_1 \leq \mu$. Obviously, $\siit$ is also positive. According to Lemma~\ref{lem:weakstar Tk implies weak op of id otimes Tk}, $\id \otimes \siit^k$ (for $k\in \Lambda$) converges in the weak operator topology to $\id \otimes \siit$, and so does $\id \otimes \siit^k - \roo^k$ to $\id \otimes \siit - \roo$. But then we can conclude that $\id \otimes \siit - \roo\geq 0$ such that $\Lambda(\roo|B) \leq  \tr\siit \leq \mu$.
\end{proof}

\begin{lem}\label{l:Lambda statement for reduction of Hmax}
For $\roo\in \mc{S}(\hahb)$, let $\rho_{ABC}$ be a purification with purifying system $H_C$, and $(P_k^A,P_k^B)$ be a generator of projected states.
 It follows that
\begin{equation} \label{prop,eq1:Lambda statement for reduction of Hmax}
\Lambda(\rho_{AC}|C) = \lim_{k \rightarrow \infty} \Lambda(\rho^k_{AC}|C),
\end{equation}
where $\rho^k_{AC} = \tr_B[(P^A_k \otimes P_k^B \otimes \id_C)\rho_{ABC}(P^A_k \otimes P_k^B  \otimes \id_C)].$
\end{lem}

\begin{proof}
Let $\nu_k := \Lambda(\rho_{AC}^k|C)$.  Due to Lemma~\ref{l:mon incr} this sequence is monotonically increasing, so we can define $\nu := \lim_{k \rightarrow \infty} \nu_k\in \mathbb{R}\cup\{+\infty\}$, and conclude that $\nu \leq \Lambda(\rho_{AC}|C)$.  Thus, the case $\nu=+\infty$ is trivial. It thus remains to show $\nu \geq \Lambda(\rho_{AC}|C)$ for $\nu <+\infty$. As proved in Lemma~\ref{l:Lambda statement for reduction of Hmin}, the infimum in Eq.~(\ref{lambdadef2}) is attained even if the underlying Hilbert spaces are infinite-dimensional. Thereby there exists for each $k \in \mathbb{N}$ a state $\sitC^k$ such that $\id\otimes\sitC^k \geq \rho_{AC}^k$ and $\tr\sitC^k =\Lambda(\rho_{AC}^k|C)$. Now we can proceed in the same manner as in the proof of Lemma~\ref{l:Lambda statement for reduction of Hmin} to construct a weak* limit $\sitC \in \tau_1^+(H_B)$ that satisfies $\id_A \otimes \sitC \geq \rho_{AC}$, and is such that  $\Lambda(\rho_{AC}|C)  \leq  \tr\sitC \leq \nu \leq \Lambda(\rho_{AC}|C)$. This completes the proof.
\end{proof}

Of course, Lemma~\ref{l=sup} and \ref{l:Lambda statement for reduction of Hmin} can directly be rewritten in terms of min-entropies and yield the first two statements of Proposition~\ref{p:reduction of Hmin to finite dim}. The part for the normalized projected states follows via $H_{\mathrm{min}}(\roh^k_{AB}|\sih^k_B) = H_{\mathrm{min}}(\rho^k_{AB}|\sigma^k_{B}) - \log\tr\sigma_B^k + \log\tr\rho_{AB}^{k}$, and $H_{\mathrm{min}} (\roh^k_{AB}|B) = H_{\mathrm{min}}(\roo^k|B) + \log\tr\roo^k $.

In order to obtain the convergence stated for the max-entropy in Proposition~\ref{p:reduction of Hmin to finite dim}, note that $(P^A_k \otimes P_k^B \otimes \id_C) \rho_{ABC} (P^A_k \otimes P_k^B \otimes \id_C)$ is a purification of $\rho_{AB}^k$, whenever $\rho_{ABC}$ is a  purification of $\rho_{AB}$. Hence, $H_{\mathrm{max}}(\rho_{AB}^k|B) = -H_{\mathrm{min}}(\rho_{AC}^k|C) = \log\Lambda(\rho_{AC}^k|C)$. For normalized states use $H_{\mathrm{max}} (\hat{\rho}_{AB}^k|B_k) = H_{\mathrm{max}}(\roo^k|B_k) - \log\tr\roo^k$.

\end{appendix}

\end{document}